\documentclass{sig-alternate}
\toappear{} % strip the ACM statement

\usepackage{latexsym}
\usepackage{epsfig}
\usepackage{cite}
\usepackage{amssymb}
\usepackage{amsmath}
\usepackage{psfrag}
\usepackage{verbatim}
\usepackage[colorlinks,urlcolor=blue,citecolor=blue,linkcolor=blue]{hyperref}
\usepackage{url}
\usepackage{graphicx}
\usepackage{epstopdf}
\graphicspath{{.}{./figures/}}
\usepackage[font=bf]{caption}
\DeclareCaptionType{copyrightbox}
\usepackage{subfig}
\usepackage{tabularx}
\usepackage{array}
\newcolumntype{U}{>{\raggedright\arraybackslash}X}
\newcolumntype{Y}{>{\raggedleft\arraybackslash}X}
\usepackage{wrapfig}
\toappear{}

\usepackage{paralist}
\usepackage{booktabs}
\usepackage{textcase}

\newtheorem{theorem}{Theorem}[section]

\newtheorem{lemma}[theorem]{Lemma}

\newtheorem{claim}[theorem]{Claim}

\newcommand{\algn}[1]{\textsc{\MakeTextLowercase{#1}}}

\newcommand{\SecLabel}[1]{\hyperref[sec:#1]{\textbf{Section \ref*{sec:#1}}.}} %section
\newcommand{\Sec}[1]{\hyperref[sec:#1]{\S\ref*{sec:#1}}} %section
\newcommand{\Eqn}[1]{\hyperref[eq:#1]{(\ref*{eq:#1})}} %equation
\newcommand{\Fig}[1]{\hyperref[fig:#1]{Figure~\ref*{fig:#1}}} %figure
\newcommand{\Tab}[1]{\hyperref[tab:#1]{Table~\ref*{tab:#1}}} %table
\newcommand{\Thm}[1]{\hyperref[thm:#1]{Theorem~\ref*{thm:#1}}} %theorem
\newcommand{\Lem}[1]{\hyperref[lem:#1]{Lem.\,\ref*{lem:#1}}} %lemma
\newcommand{\Prop}[1]{\hyperref[prop:#1]{Prop.~\ref*{prop:#1}}} %property
\newcommand{\Cor}[1]{\hyperref[cor:#1]{Cor.~\ref*{cor:#1}}} %corollary
\newcommand{\Def}[1]{\hyperref[def:#1]{Defn.~\ref*{def:#1}}} %definition
\newcommand{\Alg}[1]{\hyperref[alg:#1]{Alg.~\ref*{alg:#1}}} %algorithm
\newcommand{\Ex}[1]{\hyperref[ex:#1]{Ex.~\ref*{ex:#1}}} %example
\newcommand{\Clm}[1]{\hyperref[clm:#1]{Claim~\ref*{clm:#1}}} %example

\newcommand{\dm}{d_{max}}
\DeclareMathOperator{\vol}{\texttt{vol}}
\DeclareMathOperator{\cut}{\texttt{cut}}
\DeclareMathOperator{\edges}{\texttt{edges}}

\newcommand{\EX}{\hbox{\bf E}}

\newcommand{\ball}[2]{N_#1(#2)}

\begin{document}

\title{Neighborhoods are good communities}
%%\titlenote{A full version of this paper is available as
%%\textit{Author's Guide to Preparing ACM SIG Proceedings Using
%%\LaTeX$2_\epsilon$\ and BibTeX} at
%%\texttt{www.acm.org/eaddress.htm}}}

\numberofauthors{2} %  in this sample file, there are a *total*
\author{
\alignauthor
David F. Gleich\\
       \affaddr{Purdue University}\\
       \affaddr{Computer Science Department}\\
       \email{dgleich@purdue.edu}
% 2nd. author
\alignauthor
C. Seshadhri\thanks{The author is supported by the Sandia LDRD program (under
project 158477) and the applied mathematics program at the United
States Department of Energy.}\\
       \affaddr{Sandia National Laboratories}\thanks{Sandia National Laboratories is a multi-program laboratory managed and operated by Sandia Corporation, a wholly owned subsidiary of Lockheed Martin Corporation, for the U.S. Department of Energy's National Nuclear Security Administration under contract DE-AC04-94AL85000.}\\
       \affaddr{Livermore, CA}\\
       \email{scomand@sandia.gov}
}

\date{}

\maketitle
\begin{abstract}
The communities of a social network are sets of vertices
with more connections inside the set than outside.  
We theoretically demonstrate that two commonly
observed properties of social networks, heavy-tailed
degree distributions and large clustering coefficients,
imply the existence of vertex neighborhoods (also known
as egonets) that are themselves good communities.  We
evaluate these neighborhood communities on a range of graphs.
What we find is that the neighborhood communities often
exhibit conductance scores that are as good as the Fiedler
cut.  Also, the conductance of neighborhood communities 
shows similar behavior as the network community profile
computed with a personalized PageRank community detection
method.  The latter requires sweeping over a great many
starting vertices, which can be expensive.  By using a small and 
easy-to-compute set of neighborhood communities
as seeds for these PageRank communities, however, we find communities 
that precisely capture the behavior of the network community profile
when seeded everywhere in the graph, and at a significant reduction
in total work. 
\end{abstract}

%A category including the fourth, optional field follows...
%\category{D.2.8}{Software Engineering}{Metrics}[complexity measures, performance measures]


 % NOT required for Proceedings

\section{Introduction}

\emph{Community detection}, loosely speaking, is any
process that takes a graph or network and picks out
sets of related nodes.  An incredibly variety of 
techniques exist for this single task, which
has a variety of names as well: community detection,
graph clustering, and graph partitioning.  Throughout
this manuscript, we shall use the term community
and cluster interchangeably.  For more
information about approaches for this problem,
see the recent survey by Schaffer~\cite{Schaeffer-2007-clustering}.
In many techniques, a community is defined as a set with
a good score under a quality measure that reflects the
connectivity between the set and the rest of the network.  
Common measures are based on density of local edges,
deviance from a random null model,
 the behavior of random walks, or graph cuts.
Mostly, these measures are NP-hard to optimize.

To keep this manuscript simple, we shall evaluate communities
using their \emph{conductance store}.  Schaeffer identified
this measure as one of the most important cut-based measures
and it has been studied extensively in a variety of 
disciplines~\cite{Chung-1992-book,Shi2000-normalized-cuts,Kannan-2004-clusterings}.
Work by Leskovec et al.~has recently demonstrated that, although different quality measures produce
differences in terms of specific communities, strong communities
persist under a variety of measures~\cite{Leskovec-2010-empirical}.

\setlength{\tabcolsep}{1ex}
\begin{wrapfigure}{r}{0.4\linewidth}
\includegraphics[width=0.65\linewidth]{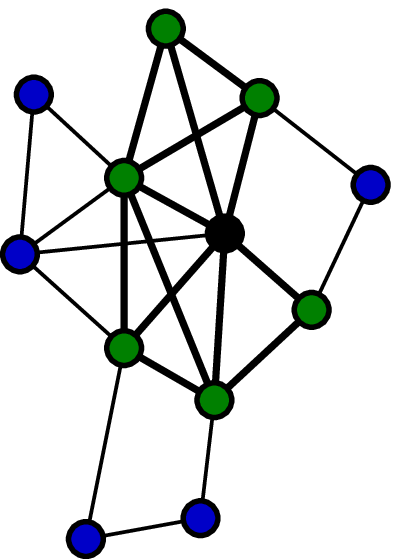}
\end{wrapfigure}
A vertex
neighborhood of a vertex $v$ is the set of vertices directly connected to 
$v$ via an edge and $v$ itself. For example, see the green and black vertices at right.
\emph{What we show here is that the presence of two commonly observed properties of 
modern information networks -- a large global clustering coefficient~\cite{watts1998-dynamics}
and a power-law degree distribution~\cite{barabasi1999-scaling} -- implies the existence of
vertex neighborhoods with good conductance scores.}   
% TODO if time - put in better figure here
We make
this statement precise in \Thm{cond}.
These results can be seen as an extension of 
the simple observation that, in the extreme case
when the global clustering coefficient of a network is 1, then the
network must be a union of cliques.  Neighborhoods define ideal communities
in this case. We mathematically show that this argument can be extended
to the case when the graph has a power-law degree distribution and a large 
clustering coefficient. 
%, neighborhoods
%will be good communities without perfect clustering.
The significance of this finding is that robust community detection
need not employ complicated algorithms.  Instead, a straightforward 
approach that just involves counting triangles -- a function that 
is easy to implement in MapReduce~\cite{Cohen-2009-graph-twiddling} 
and easy to approximate~\cite{Kolountzakis-2010-triangle-counting},
suffices to identify communities.
It is intriguing that arguably the two
most important measurable quantities of social networks imply that
communities are very easy to find. This may lead to more mathematical
work explaining the success of community detection algorithms, given
that the problem are in general NP-hard.
We note that unfortunately, our theoretical bounds reflect a worst
case behavior and are weaker than required for
practical use.  Consequently, in the remainder of the
paper we explore the utility of neighborhood communities
empirically.  

\SecLabel{formal} The technical discussion of the manuscript begins
by introducing our notation and precisely defines the quantities
we examine, such as clustering coefficients, due to variability
in the definitions of these measures.  We also discuss
the Andersen-Chung-Lang personalized PageRank clustering 
scheme~\cite{andersen2006-local} and the network whiskers
from Leskovec et al.~\cite{leskovec2008-community,Leskovec-2009-community-structure}.
We utilize the latter two algorithms as reference points for the
success of our community detection.

\SecLabel{related} We discuss some of the other observed
properties of egonets, or vertex neighborhoods, along with
other related work including overlapping communities.

\SecLabel{theory} We state and prove the theoretical
results that graphs with a power-law degree distribution
and large clustering coefficients have neighborhood
communities with good conductance scores.

\SecLabel{data} We review the data that will serve as the testbed
for our empirical evaluation of neighborhood cuts.  This comes from
a variety of public sources and spans collaboration networks,
social networks, technological networks, web networks, and 
random graph models.

\SecLabel{neigh-comm} Our empirical investigation of neighborhood
clusters takes the following form.  We first exhibit the conductance
scores for the set of neighborhood communities for a few graphs
(e.g.~\Fig{neigh-cond}).  
We find that neighborhood communities reflect the 
shape of the network community plot observed by 
Leskovec~et~al.~\cite{leskovec2008-community,Leskovec-2009-community-structure}
at small size scales.
We next compare the best neighborhood 
communities to those discovered by four other procedures: 
the Fiedler community, the best personalized PageRank community 
(\Sec{ppr}), the best network whisker (\Sec{whisker}),
and the best clusters from \algn{metis}~\cite{karypis1998-metis}.
In one third of the cases, the neighborhood community is as good
as the best of any of the other algorithms.

%We also explore how well these clusters \emph{cover} the graph,
%and discuss choosing a subset of the best clusters via a local
%process.  
Another outcome of the theory from \Sec{theory} is
that large cores must exist in these graphs.  (Here,
a graph $k$-core is a subset of vertices where all nodes
have degree at least $k$~\cite{Seidman1983-cores}.)  We conclude
this section by exploring the community properties of the 
graph $k$-cores.

\SecLabel{seeds} Motivated by the success of the neighborhood communities
at small size scales, we explore using the best vertex neighborhoods
as \emph{seeds} for a local greedy community expansion procedure
and for the Andersen-Chung-Lang algorithm.  Here, we find
that these procedures, when seeded with an easy-to-identify
set of neighborhood communities, produce larger clusters
that decay as expected by the results in 
Leskovec~et~al.~\cite{leskovec2008-community,Leskovec-2009-community-structure}

We make all of our algorithm and experimental code, the majority
of the data for the experiments, and some extra figures that did not
fit into the paper available:
\begin{center}\fontsize{8}{10}\selectfont
\url{www.cs.purdue.edu/homes/dgleich/codes/neighborhoods}
\end{center}
These codes are easy to use.  
Given the adjacency matrix of a network \verb$A$, the single command
\begin{verbatim}
>> ncpneighs(A)
\end{verbatim}
will produce a figure analyzing the neighborhood communities
in comparison to the Fiedler community (formal definition in Section~\ref{sec:fiedler}).

\paragraph{Summary of contributions}
\begin{compactitem}
\item We theoretically motivate the study
of neighborhood communities by showing
they often have a low conductance in graphs
with a power-law degree distribution and 
large clustering coefficients. 
\item We empirically evaluate these
neighborhood communities and find them
comparable to those communities found
by other algorithms at small size scales.
\item We find a small set of neighborhood
communities that can be grown into larger
communities using a PageRank based community
detection algorithm.  The results match
those communities found with a more expensive
sweep over all communities.
\end{compactitem}

\section{Formal setting and notation}\label{sec:not}\label{sec:formal}
We first list out the various notations and formalisms used.
All of the key notation is summarized
in \Tab{notation}, and we briefly review it here.
Let $G =(V,E)$ be a loop-less, undirected, unweighted graph. 
We denote the number of vertices by $n = |V|$ and
the number of edges by $m = |E|$. 
In terms of the adjacency matrix, $m$ is half the number of 
non-zeros entries.
For a vertex $v$, let $d_v$ be the degree of $v$. 
For any positive integer $d$, let $f_d$ be the
number of vertices of degree $d$, that is, 
the frequency of $d$ in the degree distribution.
The maximum degree is denoted by $d_{max}$.
Let $D_r(v)$ to be the distance $r$-neighborhood of $v$. This is the set
of vertices whose shortest path distance from $v$ is exactly $r$. 
Then, we define the ball
of distance $r$ around $v$, denoted by $\ball{r}{v}$, as the set
$\bigcup_{i \leq r} D_r(v)$. 

\begin{table}
\caption{A summary of the notation.} \label{tab:notation}
\begin{tabularx}{\linewidth}{@{}rX@{}}
\toprule
$n=|V|$ & the number of vertices\\
$m=|E|$ & the number of edges\\
$d_v$ & the degree of vertex $v$\\
$f_d$ & the number of vertices of degree $d$\\
$W$ & the set of wedges in a graph\\
$W_v$ & the set of wedges centered at vertex $v$\\
$\kappa$ & the global clustering coefficient\\
$\bar{C}$ & the mean local clustering coefficient\\
$C_v$ & the local clustering coefficient for vertex $v$\\
$\ball{r}{v}$ & the set of vertices within distance $r$ or $v$\\
$E(S,T)$ & the set of edges between $S$ and $T$\\
$\cut(S)$ & the size of the cut around vertex set $S$\\
$\vol(S)$ & the sum of degrees (volume) of vertices in $S$\\
$\edges(S)$ & twice the number of edges among vertices in $S$\\
$\phi(S)$ & the conductance of vertex set $S$\\
\bottomrule
\end{tabularx}
\end{table}

\subsection{Clustering coefficients}

A \emph{wedge} is an unordered pair of edges that share an endpoint. 
The \emph{center}
of the wedge is the common vertex between the edges. 
A wedge $\{(s,t),(s,u)\}$ is \emph{closed}
if the edge $(t,u)$ exists, and is open otherwise.
We use $W$
to denote the set of wedges in $G$, and $W_v$ for the set
of wedges centered at $V$. 
Note that $|W_v| = \binom{d_v}{2}$. We set $p_v = |W_v|/|W|$.

Social networks often have large 
\emph{clustering coefficients}~\cite{watts1998-dynamics}. 
Because
of the varying definitions of this term that are used, 
we will denote by $\kappa$
the \emph{global clustering coefficient}. This quantity
is basically a normalized count of triangles. In the following,
we think of $w$ drawn uniformly at random from $W$. 
\[ \kappa = \Pr_{w \in W}[\text{$w$ is closed}] = \frac{\text{number of closed wedges}}{|W|} \]
In terms of triangles, $ \kappa = {3 \cdot \text{number of triangles}}/{|W|}.$
For any vertex $v$, $C_v$ is the \emph{local clustering coefficient of $v$}.
We draw $w$ uniformly at random from $W_v$.
\[ C_v = \Pr_{w \in W_v}[\text{$w$ is closed}] = \frac{\text{number of closed wedges in $W_v$}}{|W_v|} \]

\subsection{Cuts and Conductance}

Given a set of vertices $S$, the set $\bar{S}$ is the complement 
set, $\bar{S} = V \backslash S$.  For disjoint sets of vertices $S, T$,
$E(S,T)$ denotes the edges between $S$ and $T$. For convenience,
we denote \emph{the size of the cut induced by a set} 
$|E(S,\overline{S})|$ by $\cut(S)$.  

The conductance of a cluster (a set of vertices) measures the probability that a one-step random walk
starting in that cluster leaves that cluster. Let
$\vol(S)$ denotes the sum of degrees of vertices
in $S$ and $\edges(S)$ denotes twice the number
of edges among vertices in $S$ so that 
\[ \edges(S) = \vol(S) - \cut(S). \]
Then the conductance of set $S$, denoted $\phi(S)$, is
\[ \phi(S)  = \frac{\cut(S)}{\min\bigl(\vol(S),\vol(\bar{S})\bigr)}. \]
Conductance is measured with respect
to the set $S$ or $\bar{S}$ with smaller volume,
and is the probability of picking an edge from the smaller
set that crosses the cut.
Because of this property, conductance is
preserved on taking complements: $\phi(S) = \phi(\bar{S})$.  For
this reason, when we refer to the number of vertices in
a set of conductance $\phi$, we always use the smaller
set $\min(|S|,|\bar{S}|)$.
\begin{figure*}
\makebox[\textwidth]{
\hss
\subfloat[\label{fig:sample:neigh} Best neighborhood\newline size=8, cut=10, $\phi$=0.15]{\includegraphics[width=1.5in]{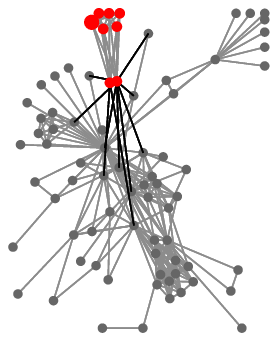}}
\hss
\subfloat[\label{fig:sample:fiedler} Fiedler community\newline size=36, cut=29, $\phi$=0.13]{\includegraphics[width=1.5in]{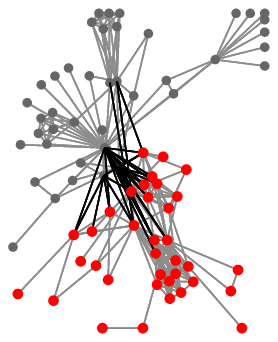}}
\hss
\subfloat[Best $k$-core\newline size=12, cut=34, $\phi$=0.22]{\includegraphics[width=1.5in]{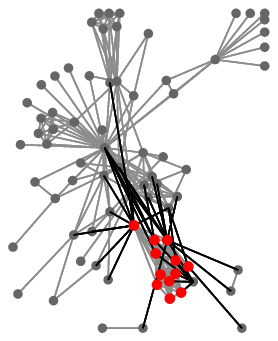}}
\hss
\subfloat[\label{fig:sample:ppr} Best \algn{PPR} community\newline size=28, cut=31, $\phi$=0.12]{\includegraphics[width=1.5in]{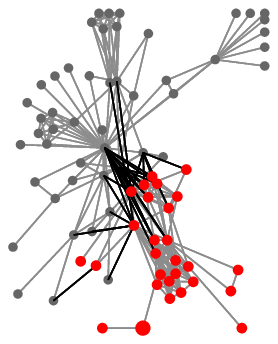}}
\hss
}
\caption{A series of vertex sets and their associated sizes and conductance 
score on the graph of characters from \emph{Les Mis\'{e}rables}~\cite{Knuth-1993-graphbase}.
The best neighborhood and best $k$-core are two of the communities we 
discuss further in \Sec{neigh-comm}.  See \Sec{methods} for information
on the Fiedler and PPR communities.
}
\label{fig:sample}
\end{figure*}
\Fig{sample} shows a few communities and their associated
cuts and conductance scores from our methods and two
points of comparison.  

\subsection{Finding good conductance communities} \label{sec:methods} 

We briefly review three ways of identifying a community
with a good conductance score.

\paragraph{Fiedler set} \label{sec:cheeger}\label{sec:fiedler}
The well-known Cheeger inequality defines 
a bound between the second smallest eigenvalue
of the normalized Laplacian matrix and the set
of smallest conductance in a graph~\cite{Chung-1992-book}.
Formally,
\[ (1/2)\lambda_2 \le \min_{S \subset V} \phi(S) \le \sqrt{2 \lambda_2} \]
where $\lambda_2$ is the second smallest eigenvalue of the
normalized Laplacian.
The proof is constructive.  It identifies
a set of vertices that obeys the upper-bound using
a \emph{sweep} cut. This is the smallest
conductance cut among all cuts
induced by ordering vertices by increasing
values of $\sqrt{d_v} x_v$, where $x_v$ is the component 
of the eigenvector associated with $\lambda_2$.
This is the same idea used in normalized cut
procedures~\cite{Shi2000-normalized-cuts}.
We refer to the set identified by this procedure
as the \emph{Cheeger} community or \emph{Fiedler}
community.  The latter term is based on Fiedler's
work in using the second smallest eigenvalue
of the combinatorial Laplacian 
matrix~\cite{Fiedler1973-algebraic-connectivity}.
\Fig{sample:fiedler} shows the Fiedler community
for the Les Mis\'{e}rables network.

\paragraph{Personalized PageRank communities} \label{sec:ppr}
Another highly successful scheme for
community detection based on conductance 
uses personalized PageRank vectors.
A personalized PageRank vector is the 
stationary distribution of a random walk
that follows an edge
of the graph  with probability $\alpha$  and 
``teleports'' back to a fixed seed vertex with probability $1-\alpha$.
We use $\alpha=0.99$ in all experiments.
The essence of the induced community 
is that an inexact personalized PageRank vector, 
computed via an algorithm
that ``pushes'' rank round the graph, will
identify good bottlenecks nearby a seed 
vertex.  These bottlenecks can be formalized
in a Cheeger-like bound~\cite{andersen2006-local}.
The procedure to find a personalized PageRank community
is: i) specify a value of $\alpha$, a seed vertex $v$,
and a desired clusted size $\sigma$; ii) solve the personalized
PageRank problem using the algorithm from~\cite{andersen2006-local}
until a degree-weighted tolerance of $\tau = 1/(10\sigma)$;
and iii) sweep over all cuts induced by the ordering of the
personalized PageRank vector (normalized by degrees) and 
choose the best.
Personalized PageRank communities (\algn{PPR} communities, for short)
were used
to identify an interesting empirical property
of communities in large networks~\cite{leskovec2008-community,Leskovec-2009-community-structure}.
To generate these plots, those authors examined
a range of values of $\sigma$ for a large number of
vertices of the graph and summarized the best communities
found at any size scale in a network community plot.
\Fig{sample:ppr} shows the best personalized PageRank
community for the network of character interactions in Les Mis\'{e}rables.

\paragraph{Whisker communities} \label{sec:whisker}
Perhaps the best point of comparison with our
approach are the \emph{whisker communities}
defined by Leskovec et al.~\cite{leskovec2008-community,Leskovec-2009-community-structure}.
These communities are small dense subgraphs connected
by a single edge.  They can be found by looking
at any subgraph connected to the largest biconnected
component by a single edge. A biconnected component
remains connected after the removal of any vertex.
Note that the largest biconnected component is not necessarily
a 2-core of the graph.  Leskovec et~al.~observed
that many of these subgraphs are rather dense.  Each
subgraph has a cut of exactly one, and consequently, 
a productive means of finding sets with low conductance
is to sort these subgraphs by their volume.  The best
whisker cut is the single subgraph with largest volume.

\section{Related work} \label{sec:related}

We are hardly the first to notice that vertex
neighborhoods have special properties.  

\paragraph{Egonets, homophily, and structural holes}
In the context of social networks, vertex neighborhoods
are often called egonets because they reflect the 
the state of the network as perceived by a single
vertex.  Their analysis is a key component 
in the study of social networks~\cite{Wasserman-1994-sna},
especially in terms of data collection.
Studies of these networks often focus on the
theory of structural holes, which is the notion that an 
individual can derive an advantage from serving
as a bridge between disparate groups~\cite{Burt-1995-structural}.
These bridge roles are interesting because they contradict
homophily in social ties.  Homophily, or the principle
that similar individuals form ties, is the mechanism that is expected to produce
networks with large local clustering coefficients~\cite{McPherson-2001-homophily}.
These social theories have prompted the development of new methods 
to tease apart some of these effects in real-world networks~\cite{LaFond-2010-homophily},
and to develop network models that capture structural holes~\cite{Kleinberg-2008-structural-holes}.

\paragraph{Clustering and communities}

Vertex neighborhoods often play a role in other
techniques to find community or clustering
structure in a network.  Overlap in the neighborhood sets of vertices
is a common vertex similarity metric used to 
guide graph clustering 
algorithms~\cite{Schaeffer-2007-clustering}.
Other schemes utilize vertex neighborhoods 
as good \emph{seed sets} for local techniques to
grow communities~\cite{Schaeffer-2006-thesis,Huang-2011-community}.
We explore using a carefully
chosen set of neighborhoods for this purpose in our
final empirical discussion (\Sec{seeds}).
Perhaps the most closely related work is a recent idea
to utilize the connected components of ego-nets, after their
ego vertex is removed, to produce a good set of overlapping
communities~\cite{Rees-2010-overlapping}.  Our
theoretical results establish that these ideas are highly
likely to succeed in networks with local clustering
and power-law degree distributions.

\paragraph{Graph properties}
Much of the modern work on networks rests on
surprising empirical observations about the
structure of real world connections.  For instance,
information networks were found to have a 
power-law in the degree 
distribution~\cite{barabasi1999-scaling,Faloutsos-1999-power-law}.
These same networks were also found to have 
considerable local structure in the form of
large clustering coefficients~\cite{watts1998-dynamics},
but retained a small global diameter.  
Our theory shows that a third potential observation -- 
the existence of vertex neighborhood with low conductance --
is in fact implied by these other two properties.
We formally show that heavy tailed degree distributions and high clustering
coefficients imply the existence of large dense cores.

\paragraph{Anomoly detection}
Predictable behavior in the structure of ego-nets
makes them a useful tool for detecting anomalous
patterns in the structure of the network.  
For instance, Akoglu et~al.~\cite{Akoglu-2010-oddball}
compute a small collection of measures on each egonet,
such as the average degree and largest eigenvalues.
Outliers in this space of vertices are often rather
anomalous vertices.  
Our work is, in contrast, a precise statement about
the regularity of the ego-nets, and says that we
always expect a large ego-net to be a good community.

\paragraph{Summary} Although we are not the first 
to study neighborhood based communities, the relationship
between the local clustering, power-law degree distributions,
and large neighborhoods with small conductance does not appear
to have been noticed before.

\section{\texorpdfstring{Theoretical justification for \\neighborhood communities}{%
Theoretical justification for neighborhood communities}} \label{sec:theory}

The aim of this section is to provide some mathematical justification for the success
of neighborhood cuts. 
%This will not be a complete proof, but for certain ranges
%of clustering coefficient values, the argument formally carries through. Outside
%those ranges, it provides intuition behind why neighborhoods form good communities.
%
Our aim is to show that heavy tailed degree distributions and large clustering coefficients
imply the existence of neighborhood cuts with low conductance and large dense cores. 
As mentioned earlier, the exact bounds we get are somewhat weak and only hold
when the clustering coefficient is extremely large. Nonetheless, the proofs give significant intuition
into \emph{why} neighborhoods are good communities. 

We begin with the extreme case when
the value of $\kappa$ is $1$ (so \emph{every} wedge is closed). Then
we have the following simple claim.

\begin{claim} \label{clm:clique} Suppose the global clustering coefficient
of $G$ is $1$. Then $G$ is the union of disjoint cliques.
\end{claim}

\begin{proof} Consider two vertices $u$ and $v$ that are connected.
Suppose the shortest path distance between them is $\ell > 1$. Then the shortest
path has at least $3$ distinct vertices (including $u$ and $v$). Take the last
three vertices on this path, $v_1, v_2, v$. This forms a wedge at $v_2$, and must
be closed (since the clustering coefficient is $1$). Hence, the edge $(v_1,v)$ exists
and there exists a path between $u$ and $v$ of length less than $\ell$. This
is a contradiction.

Hence, any two connected vertices have a shortest path distance of $1$, i.e., are connected
by an edge. The graph is a disjoint union of cliques.
\end{proof}

Note that the neighborhood of any vertex in the above claim forms a clique disconnected
from the rest of $G$. Therefore, all neighborhoods form perfect communities, in this
extremely degenerate case. We prove this for more general settings.
%The crucial aspect of the proof is to construct a distribution over the vertices,
%based on the number of wedges incident to each vertex. 
The quantities $p_v = |W_v|/|W|$,
form a distribution over the set of vertices. Since
we are performing an asymptotic analysis, we will use $o(1)$ to denote
any quantity that becomes negligible as the graph size increases.
We will choose $\beta$ to be a constant less than $1$. It is quite unimportant
for the asymptotic analysis what this constant is. From a pratical standpoint,
think of $\beta$ as a constant such that most edges are incident to
a vertex of degree at least $\dm^\beta$ ($2/3$ is usually a reasonable value).
Also, we will assume that the power law exponent is at most $3$, a fairly
acceptable condition.
%
%Practically, $\beta$ is probably
%around $2/3$.

\begin{claim} \label{clm:degdist} Let $S$ be the set of vertices with
degrees more than $\dm^{\beta}$. Then,
$ \sum_{v \in S} p_v = 1 - o(1) $.
\end{claim}

\begin{proof} We can set $p_v = (2|W_v|)/(2|W|)$.
For convenience, set $d_1 = \dm^\beta$ and $d_2 = \dm$. 
We have $f_d \approx \alpha n/d^\gamma$, for some constant $\alpha$ and $\gamma < 3$.
\[ \sum_{v \in S} 2|W_v| \approx \sum_{d = d_1}^{d_2} d^2f_d \approx \alpha n \sum_{d = d_1}^{d_2} d^{2-\gamma} \approx \alpha' n(d^{3-\gamma}_1 - d^{3-\gamma}_2) \]
The total number of wedges behaves like $\alpha' nd^{3-\gamma}_1$
and hence, $2\sum_{v \in S}|W_v| = 2|W| - o(|W|)$.
\qed\\
\end{proof}

\begin{claim} \label{clm:cc} $\sum_v p_v C_v = \kappa$
\end{claim}

\begin{proof}
\begin{eqnarray*}
\sum_v p_v C_v & = & \sum_v \frac{|W_v|}{|W|} \cdot  \frac{\textrm{number of closed wedges in $W_v$}}{|W_v|} \\
& = & \frac{\sum_v\textrm{(\# closed wedges in $W_v$)}}{|W|} = \kappa. \qed
\end{eqnarray*}
%\\
\end{proof}

We come to our important lemma. This argues that on the average, neighborhood cuts must
have a low conductance.

%below
%that when $\kappa$ is very large (but not exactly $1$), neighborhood cuts are still
%quite good. Observe that we are \emph{only} assuming the clustering coefficient is large,
%and nothing else about the graph structure. 
%
\begin{lemma} \label{lem:nbd}
$$ \sum_v \left(p_v  \frac{\cut(N_1(v))}{|W_v|}\right) = 2(1 - \kappa) $$
\end{lemma}

\begin{proof}  We express the sum of $\cut(N_1(v))$ as a double summation, and perform
some algebraic manipulations.
\begin{eqnarray*} \sum_v \cut(N_1(v)) & = & \sum_v \sum_{u \in N_1(v)} |N_1(u) \setminus (N_1(v) \cup \{v\})| \\
& = & \sum_u \sum_{v \in N_1(u)} |N_1(u) \setminus (N_1(v) \cup \{v\})| \\
& = & \sum_u \sum_{v \in N_1(u)} \textrm{(\# open wedges centered} \\
& & \ \ \ \ \textrm{at $u$ involving edge $(u,v)$)} \\
& = & 2\sum_u \textrm{(\# open wedges centered at $u$)} \\
& = & 2(1-\kappa)|W|
\end{eqnarray*}
We complete the proof with the following simple observation:
\[ \sum_v \left(p_v \frac{\cut(N_1(v))}{|W_v|}\right) = \frac{\sum_v \cut(N_1(v))}{|W|}. \qed \]
\end{proof}

%Since $p_v$'s form a distribution of the vertices, the quantity in Lemma \ref{lem:nbd} is 
%a weighted average of $\cut(N_1(v))/|W_v|$. Using a simple manipulation,
%we can look at the global clustering coefficient as a weighted average
%of local clustering coefficients.

%
%
%Let us now make an informal argument. The numbers $p_v$ form
%a probability distribution. Suppose we picked a random vertex $v$ according to this distribution.
%By Claim~\ref{clm:cc}, the expected clustering coefficient of $v$ is $\kappa$. By Claim~\ref{lem:nbd},
%the expected value of ${\cut(N_1(v))}/{|W_v|}$ is $2(1-\kappa)$. Furthermore,
%suppose most of the wedges are centered at high degree vertices. This can be formally shown for
%power law degree distributions, where $\gamma < 2$.

\begin{theorem} \label{thm:core} There exists a $k$-core in $G$
for $k \geq \kappa \dm^{\beta}/2$.
\end{theorem}

\begin{proof} By Claims~\ref{clm:degdist} and \ref{clm:cc},
$$ \kappa = \sum_v p_v C_v = \sum_{v \in S} p_v C_v + \sum_{v \in \overline{S}} p_v C_v \leq \sum_{v \in S} p_v C_v + o(1) $$
This implies that there exists some vertex $v$ such that $d_v > \dm^{\beta}$ and $C_v \geq \kappa - o(1)$ (for convenience,
we are going to drop the $o(1)$ lower order term). Consider $G'$, the induced subgraph of $G$ on $N_1(v)$. 
The total number of vertices is exactly $d_v + 1$. Because a $\kappa$-fraction of the wedges
centered at $v$ are closed, the number of edges in $G'$ is at least $\kappa \binom{d_v}{2}$.
So $G'$ is a dense graph, and we will show that it contains a large core. Perform
a core decomposition on $G'$. We iteratively remove the vertex of min-degree until
the graph has no edges left. The total number of iterations is atmost $d_v$.
Let the degree of the removed vertex at iteration $i$ be $e_i$. We
have $\sum_{1 \leq i \leq d_v} e_i = \kappa {d_v\choose 2}$. By an averaging
argument, there exists some $i$ such that $e_i \geq \kappa (d_v-1)/2$.
At this point, \emph{all} (unremoved) vertices of $G'$ must have
a degree of at least $(d_v-1)/2$, forming a $k$-core with $k \geq \kappa \dm^{\beta}/2$.
\qed
\end{proof}

We come to our main theorem that proves the existence of a neighborhood cut with low conductance. 
When $\kappa = 1$, we get back the statement of Claim~\ref{clm:clique}, since we have a set
of conductance $0$. But this theorem also gives non-trivial bounds for large values of $\kappa$.
As we mentioned earlier, when $\kappa$ becomes small, this bound is not useful any longer.

\begin{theorem} \label{thm:cond} There exists a neighborhood cut with conductance at least $4(1-\kappa)/(3-2\kappa)$.
\end{theorem}

\begin{proof} The proof uses the probabilistic method, given the bounds
of Lemma~\ref{lem:nbd} and Claim~\ref{clm:cc}. Suppose we choose a vertex $v$
according to the probability distribution given by $p_v$. Let $X$
denote the random variable ${\cut(N_1(v))}/{|W_v|}$, so $\EX[X] = 2(1-\kappa)$ (Lemma~\ref{lem:nbd}).
By Markov's inequality,
% (Theorem 3.2 of \cite{MoRa01}), 
$\Pr[X > 4(1-\kappa)] \leq 1/2$.

Set $\alpha = 2\kappa - 1$, and set $\Pr[C_v < \alpha] = p$. 
$$ \kappa < p\alpha + (1-p) \Longrightarrow p < (1-\kappa)/(1-\alpha) = 1/2 $$
By the union bound, the probability that ${\cut(N_1(v))}/{|W_v|} > 4(1-\kappa)$
\emph{or} $C_v < \alpha$ is less than $1$. Hence, there exists some vertex $v$
such that $\cut(N_1(v)) \leq 4(1-\kappa)|W_v|$ and $C_v \geq \alpha$ (we can
also show that $d_v \geq n^\beta$). Let $E$ be the set of edges in the subgraph induced on $N_1(v)$.
Since $C_v \geq \alpha$, $|E| \geq \alpha|W_v|$.
We can bound the conductance of $N_1(v)$,
\begin{eqnarray*}
\frac{\cut(N_1(v))}{|E| + \cut(N_1(v))} & \leq & \frac{4(1-\kappa)|W_v|}{\alpha|W_v| + 4(1-\kappa)|W_v|} 
= \frac{4-4\kappa}{3-2\kappa}. \qed
\end{eqnarray*}
\end{proof}

%Hence, if we randomly choose a vertex according to $p_v$, then $v$
%is a high degree vertex.
%All the arguments above suggest that there is a high degree
%vertex $v$ such that both the above statements hold. So $C_v \geq \kappa$
%and $\cut(N_1(v)) \leq 2(1-\kappa)|W_v|$. Consider the induced subgraph
%on $N_1(v)$. Since $CC_v \geq \kappa$, the number of edges in this subgraph 
%(denote this by $E_v$) is at least $\kappa |W_v|$. Note that this means that
%this subgraph is dense. (This immediately implies the existence
%of large cores.) Furthermore,
%%
%$$ \cut(N_1(v)) \leq [2(1-\kappa)/\kappa] |E_v| $$
%%
%Suppose $\kappa$ is really large. This implies that the number
%of edges \emph{leaving} $N_1(v)$ is a small fraction of the edges internal
%to $N_1(v)$. That suggests that $N_1(v)$ is a community, or a community like
%structure. 

\section{Data} \label{sec:data}

Before we begin our empirical comparison, we first discuss
the data we use to compare and evaluate algorithms.  These
come from a variety of sources.  See \Tab{data} for a summary
of the networks and their basic statistics. All networks are undirected and were symmetrized if the original
data were directed.  Also, any self-loops in the networks were
discarded.  We only look at the largest connected component
of the network.
There are five types of networks:

\begin{table}
\caption{Datasets for our experiments.  The five types are:
collaboration networks, social networks, technological networks,
web graphs, and forest fire models.}
\label{tab:data}
\fontsize{8}{11}\selectfont
\begin{tabularx}{\linewidth}{@{}l@{\,\;}l@{\;\;\,}l@{\;\;}U@{\;}U@{\;}U@{\;}U@{}}
\toprule
Graph & Verts & Edges & Avg. Deg. & Max Deg. & $\kappa$ & $\bar{C}$ \\
\midrule
       ca-AstroPh &   17903 &   196972 &  22.0 &   504 & 0.318 & 0.633 \\ 
      email-Enron &   33696 &   180811 &  10.7 &  1383 & 0.085 & 0.509 \\ 
    cond-mat-2005 &   36458 &   171735 &   9.4 &   278 & 0.243 & 0.657 \\ 
           arxiv  &   86376 &   517563 &  12.0 &  1253 & 0.560 & 0.678 \\ 
             dblp &  226413 &   716460 &   6.3 &   238 & 0.383 & 0.635 \\ 
   hollywood-2009 & 1069126 & 56306653 & 105.3 & 11467 & 0.310 & 0.766 \\ 
\midrule 
    fb-Penn94 &   41536 &  1362220 &  65.6 &  4410 & 0.098 & 0.212 \\ 
fb-A-oneyear & 1138557 &  4404989 &   7.7 &   695 & 0.038 & 0.060 \\ 
 fb-A & 3097165 & 23667394 &  15.3 &  4915 & 0.048 & 0.097 \\ 
    soc-LiveJournal1 & 4843953 & 42845684 &  17.7 & 20333 & 0.118 & 0.274 \\ 
\midrule 
      oregon2-010526 &   11461 &    32730 &   5.7 &  2432 & 0.037 & 0.352 \\ 
      p2p-Gnutella25 &   22663 &    54693 &   4.8 &    66 & 0.005 & 0.005 \\ 
         as-22july06 &   22963 &    48436 &   4.2 &  2390 & 0.011 & 0.230 \\ 
         itdk0304 &  190914 &   607610 &   6.4 &  1071 & 0.061 & 0.158 \\ 
\midrule 
          web-Google &  855802 &  4291352 &  10.0 &  6332 & 0.055 & 0.519 \\ 
\midrule
%     hyper &    3483 &    11381 &   6.5 &  2300 & 0.018 & 0.667 \\ 
   ff-0.4 &   25000 &    56071 &   4.5 &   112 & 0.283 & 0.412 \\ 
  ff-0.49 &   25000 &   254180 &  20.3 &  1722 & 0.148 & 0.447 \\ 
\bottomrule
\end{tabularx}
\end{table}

\textbf{Collaboration networks} In these networks, the nodes
represent people.  The edges represent collaborations,
either via a scientific publication (ca-AstroPh~\cite{Leskovec-2007-densification},
 cond-mat-2005~\cite{Newman-2001-collaboration},
arxiv~\cite{Bonchi-2011-fast-katz}, dblp~\cite{Boldi-2011-layered,boldi2004-webgraph}),
 an email (email-Enron~\cite{Leskovec-2009-community-structure}), 
or a movie (hollywood-2009~\cite{Boldi-2011-layered,boldi2004-webgraph}).
These networks have large mean clustering coefficients and large global
clustering coefficients.

\textbf{Social networks} The nodes are people again, and
the edges are either explicit ``friend'' relationships 
(fb-Penn94~\cite{Traud-2011-facebook}, fb-A~\cite{Wilson-2009-social-networks}, 
soc-LiveJournal~\cite{Backstrom-2006-group-formation})
 or
observed network activity over edges in a one-year
span (fb-A-oneyear~\cite{Wilson-2009-social-networks}).

\textbf{Technological networks} 
The nodes act in a distributed communication network
either as agents (p2p-Gnutella25~\cite{Matei-2002-gnutella}) or as routers
(oregon2~\cite{Leskovec-2007-densification}, 
as-22july06~\cite{Newman-2006-network}, 
itdk0304~\cite{Caida-network}).  The edges are 
observed communications between the nodes.

\textbf{Web graphs} The nodes are web-pages, and the 
edges are symmetrized links between the 
pages~\cite{Leskovec-2009-community-structure}.

\textbf{Forest fire models} We also explore 
the forest fire graph model~\cite{Leskovec-2007-densification}.
This model has large clustering coefficients and a highly
skewed degree distribution.  The model grows a network
by adding a node at each step. On arrival, a new node
picks a template uniformly
at random from the existing nodes, and then the process
``burns'' around that node with a specified probability.
Burned nodes are then connected to the new node.  It has three parameters:
the size of the initial clique $k$, the probability of following
an edge in the burning process $p$, and the total number of nodes
$n$.  We specify $k=2$ and $n=25000$, and explore two choices for $p$:
short-burning $p=0.4$ and long-burning $p=0.49$.  
%In prior
%empirical investigations, we found that forest fire models
%were the most similar to social networks among a range
%of different models.

\section{\texorpdfstring{Empirical Neighborhood\\ Communities}{Empirical Neighborhood Communities}}
\label{sec:neigh-comm}

To compute the conductance scores for each neighborhood in the graph,
we adapt any procedure to compute all local clustering coefficients.
Most of the work to compute a local clustering coefficient is performed
when finding the number of triangles at the vertex.  We can 
express the number of triangles as $\edges(D_1(v))/2 = (\edges(\ball{1}{v})/2 - 2d_v$,
that is, half the number of edges between immediate neighbors of $v$ (recall
that we double-count edges).
Then $\cut(\ball{1}{v}) = \vol(\ball{1}{v}) - \edges(\ball{1}{v})$.
And so, given the number of triangles, we can compute the cut
assuming we can compute the volume
of the neighborhood.  This is easy to do with any graph structure
that explicitly stores the degrees.  We also note that it's easy
to modify Cohen's procedure for computing triangles with 
MapReduce~\cite{Cohen-2009-graph-twiddling} 
to compute neighborhood conductance scores.
Two extra steps are required: i) map each triangle back to
its constituent nodes, then reduce to find the number of triangles
at each node; and ii) map the joined edge and degree graph
to both vertices in the edge, then sum the degrees of 
the neighborhood in the reduce.

We use the network community plot from 
Leskovec~et~al.~\cite{leskovec2008-community} to show the information on all of the neighborhood communities.
Given the conductance scores from all the neighborhood communities
and their size in terms of number of vertices, we first identify the 
best community at each size.  The network community plot shows
the relationship between best community conductance and community
size on a log-log scale.   In
Leskovec~et~al., they found that these plots had
a characteristic shape for modern information networks:
an initial sharp decrease until the community size reaches
between 100 and 1000, then a considerable rise in the conductance
scores for larger communities.   In our case, neighborhood
communities cannot be any larger than the maximum degree plus
one, and so we mark this point on the graphs.  We always
look at the smaller side of the cut, so no community
can be larger than half the vertices of the graph.  We also
mark this location on the plots.  Each subsequent
figure utilizes this size-vs-conductance plot. Note that we
deliberately attempt to preserve the axes limits across figures to promote
comparisons.  However, some of the figures \emph{do have
different axis limits} to emphasize the range of data.

First, we show these network community plots, or perhaps better 
termed neighborhood community plots for our purposes,
for six of the networks in \Fig{neigh-cond}.  These
figures are representative of the best and worst of
our results.  As a reminder, we make all summary
data and codes available online.  Plots for other
graphs are available on the website given in the introduction.

\begin{figure}[t]
\fontsize{8}{10}\selectfont
\begin{tabular}{@{}c@{}c@{}}
web-Google & itdk0304\\[-0.5ex]
\includegraphics[width=0.5\linewidth]{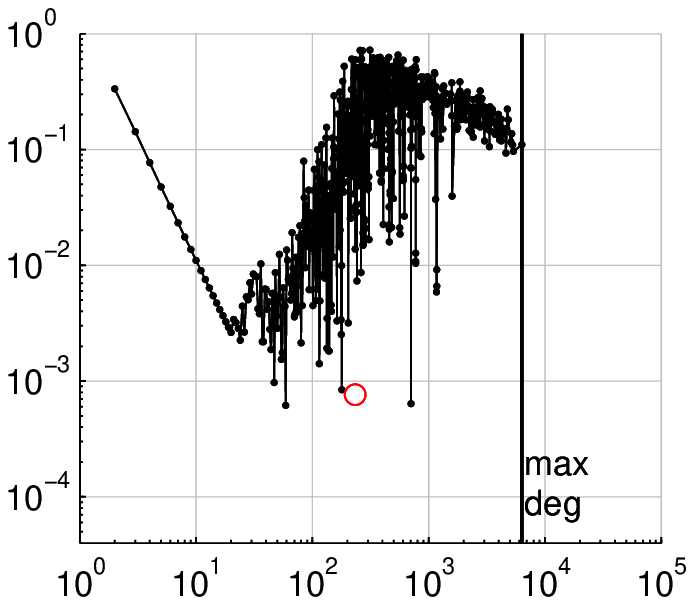} 
& \includegraphics[width=0.5\linewidth]{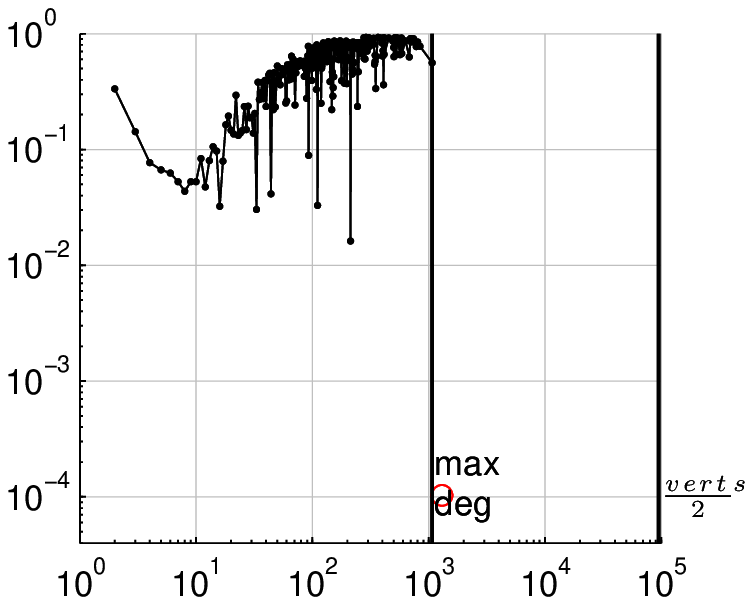} \\
fb-A-oneyear & arxiv\\[-0.5ex]
\includegraphics[width=0.5\linewidth]{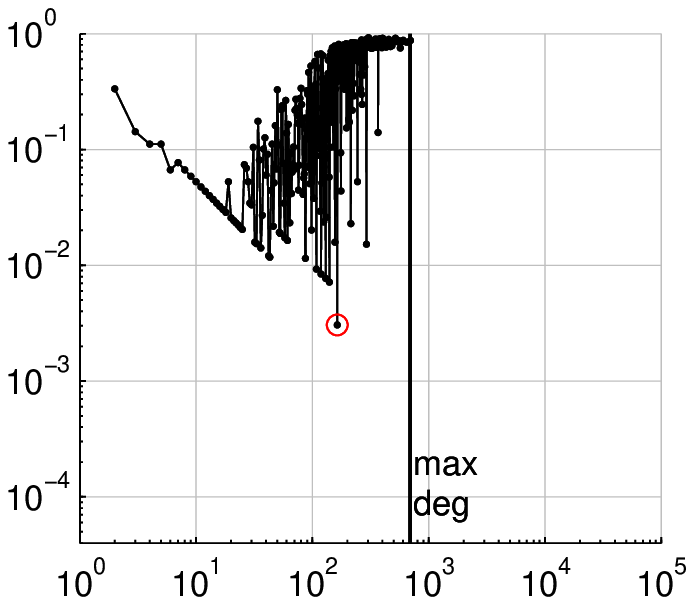} 
& \includegraphics[width=0.5\linewidth]{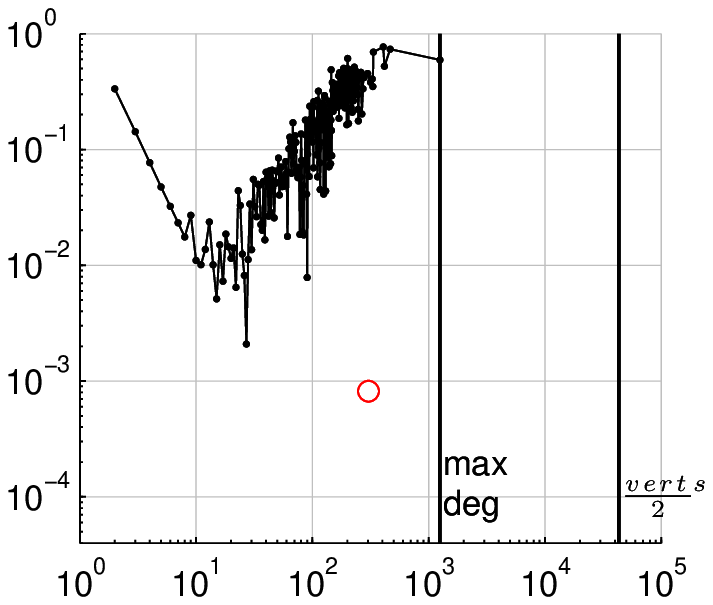} \\
soc-LiveJournal1 & ca-AstroPh\\[-0.5ex]
\includegraphics[width=0.5\linewidth]{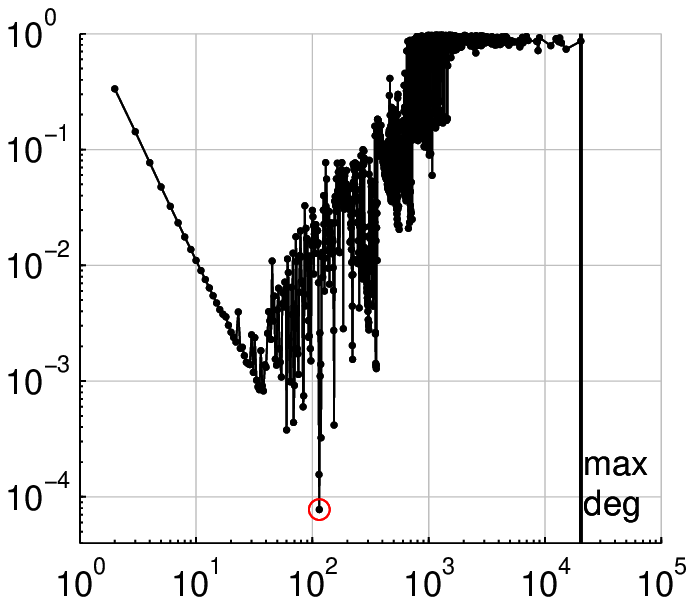} 
& \includegraphics[width=0.5\linewidth]{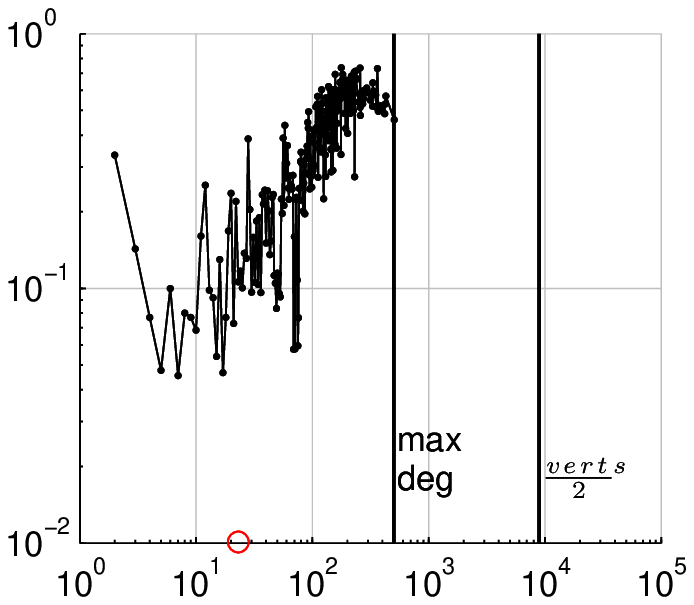} \\
Number of vertices in cluster & Number of vertices in cluster \\
\end{tabular}
\vspace{-7pt}
\caption{The best neighborhood community conductance at each size 
(black) and the Fiedler community (red).  (Note the axis limits on ca-AstroPh).}
%the Fiedler community is the red dot.}
\label{fig:neigh-cond}
\end{figure}

The three graphs on the left show cases where a neighborhood
community \emph{is or is nearby} the best Fiedler community 
(the red circle).  The three graphs on the right highlight
instances where the Fiedler community is much better than
any neighborhood community.  We find it mildly surprising
that these neighborhood communities can be as good as the 
Fiedler community.  The structure of the plot for both
fb-A-oneyear and soc-LiveJournal1 is instructive.  Neighborhoods
of the \emph{highest degree} vertices are not community-like --
suggesting that these nodes are somehow exceptional.  In fact,
by inspection of these communities, many of them are nearly
a star graph.  However, a few of the large degree nodes define
strikingly good communities (these are sets with a few
hundred vertices with conductance scores of around $10^{-2}$).  
This evidence concurs with the intuition from \Thm{cond}.  

Note that all of these plots
show the same shape Leskovec~et~al.~\cite{leskovec2008-community}
observed.  Consequently, in the next set of figures,
and in the remainder of the empirical investigation, we compare 
our neighborhood communities against those
computed via the personalized PageRank community scheme
employed in that work and described in Section~\ref{sec:methods}.

Second, \Fig{nvsncp} compares the neighborhood communities
to those computed by sweeping the local personalized PageRank
algorithm over all of the vertices as described by 
Leskovec~et~al.~\cite{leskovec2008-community}.  We also
show the behavior of the whisker communities in this plot as well.
The plot adopts the same style of figure.  The PageRank
communities are in a deep blue color, and the whisker communities
are show in a shade of green.  Here, we see that the
neighborhood communities show similar behavior at small
size scales (less than 20 vertices), but the personalized
PageRank algorithm is able to find larger communities
of smaller, or similar conductance.  In these four
cases (which are representative of all of the remaining 
figures), one of the personalized PageRank communities
was the Fiedler community.

\begin{figure}
\fontsize{8}{10}\selectfont
\begin{tabular}{@{}c@{}c@{}}
cond-mat-2005 & email-Enron\\[-0.5ex]
\includegraphics[width=0.5\linewidth]{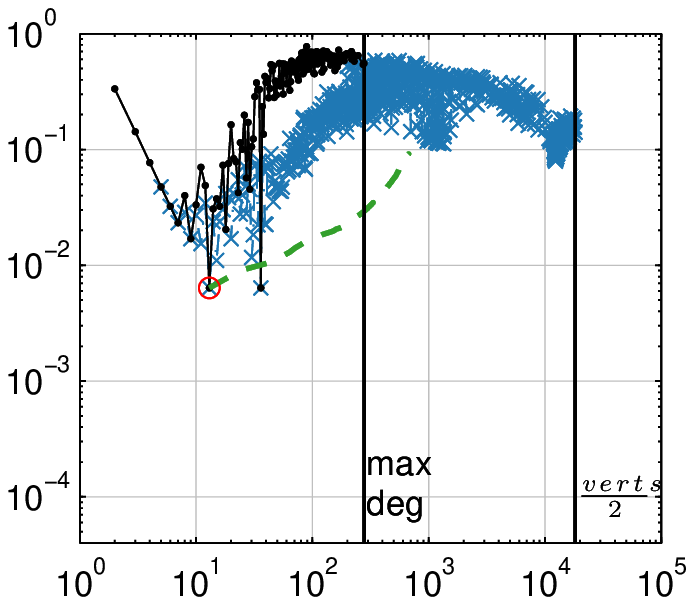} 
& \includegraphics[width=0.5\linewidth]{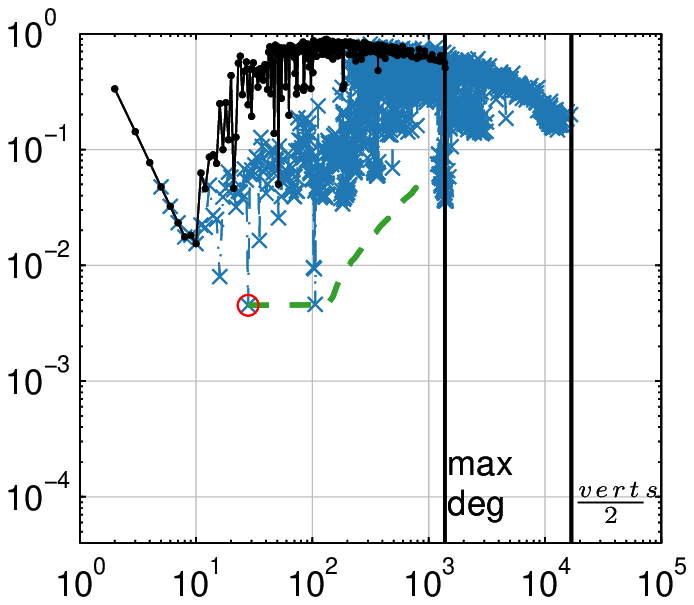} \\
hollywood-2009 & as-22july06\\[-0.5ex]
\includegraphics[width=0.5\linewidth]{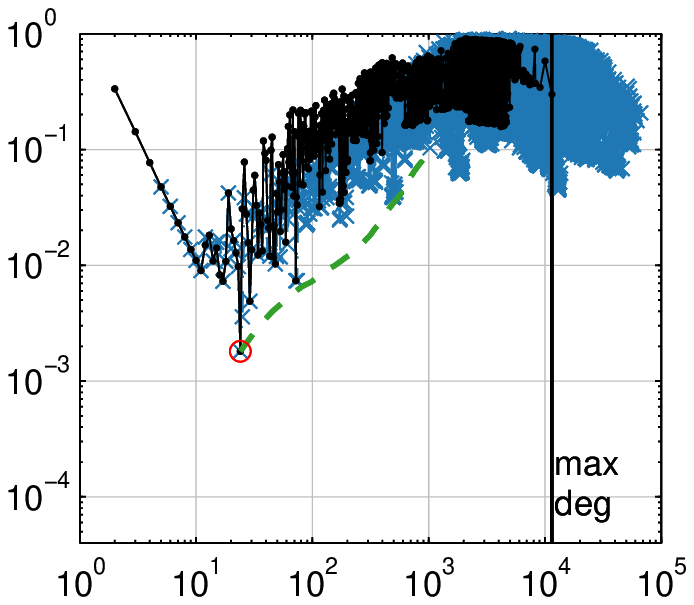} 
& \includegraphics[width=0.5\linewidth]{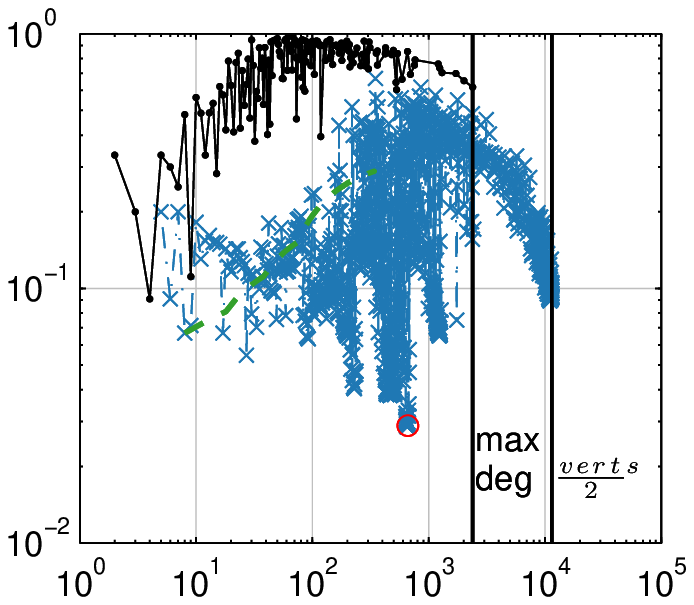} 
\end{tabular}
\vspace{-7pt}
\caption{A comparison of neighborhood communities (black)
personalized PageRank communities (blue),
and whiskers (green).}
\label{fig:nvsncp}
\end{figure}

Based on this observation, we wanted to understand how
the best community identified by a range of algorithms
compares to the neighborhood communities.
This is what our third exploration does.  The results are shown in \Tab{best}.
We computed a set of communities with \algn{metis} by repeatedly
calling the algorithm, asking it to use \emph{more} partitions
each time.  See our online codes for the precise details of 
which partitions were used.

\begin{table*}
\caption{The single best community detected by any of the five methods explore.}
\label{tab:best}
 \begin{tabularx}{\linewidth}{lYXYXYXYXYX}
\toprule
 \textbf{Graph} & \multicolumn{2}{c}{\textbf{Neighborhood}} & \multicolumn{2}{c}{\textbf{Fiedler}} 
   & \multicolumn{2}{c}{\textbf{PageRank}} & \multicolumn{2}{c}{\textbf{Whisker}} & \multicolumn{2}{c}{\textbf{Metis}} \\
& Cond\rlap{.} & Size & Cond\rlap{.} & Size & Cond\rlap{.} & Size & Cond\rlap{.} & Size & Cond\rlap{.} & Size\\
\midrule
       ca-AstroPh & 0.0455 &     7 & 0.0101 &    23 & 0.0101 &    23 & 0.0101 &    23 & 0.0101 &    23 \\
      email-Enron & 0.0154 &    10 & 0.0045 &    28 & 0.0045 &    28 & 0.0045 &    28 & 0.0080 &    16 \\
cond-mat-2005 & 0.0064 &    13 & 0.0064 &    13 & 0.0064 &    13 & 0.0064 &    13 & 0.0154 &    11 \\
           arxiv & 0.0021 &    27 & 0.0008 &   303 & 0.0014 &   304 & 0.0021 &    27 & 0.0021 &    27 \\
             dblp & 0.0038 &    24 & 0.0038 &    25 & 0.0034 &    83 & 0.0038 &    25 & 0.0041 &    17 \\
   hollywood-2009 & 0.0018 &    24 & 0.0018 &    24 & 0.0018 &    24 & 0.0018 &    24 & 0.0018 &    24 \\
\midrule 
              Penn94 & 0.3333 &     2 & 0.1898 &  7191 & 0.1966 &    41 & 0.3333 &     2 & 0.1986 &  6923 \\
fb-A-oneyear & 0.0031 &   164 & 0.0031 &   164 & 0.0031 &   164 & 0.0031 &   164 & 0.0090 &    56 \\
 fb-A & 0.0345 &     8 & 0.0084 &   647 & 0.0084 &   647 & 0.0133 &    38 & 0.0130 &    77 \\
    soc-LiveJournal1 & 0.0001 &   115 & 0.0001 &   115 & 0.0001 &   115 & 0.0001 &   115 & 0.0001 &   115 \\
\midrule 
      oregon2-010526 & 0.1368 &    12 & 0.0467 &   316 & 0.0438 &   318 & 0.1429 &     4 & 0.0553 &  3820 \\
      p2p-Gnutella25 & 0.1429 &    10 & 0.0417 &    24 & 0.0417 &    24 & 0.0588 &     9 & 0.0417 &    24 \\
         as-22july06 & 0.0909 &     4 & 0.0289 &   661 & 0.0286 &   654 & 0.0667 &     8 & 0.0296 &   657 \\
         itdk0304 & 0.0162 &   213 & 0.0001 &  1306 & 0.0002 &  1188 & 0.0001 &  1306 & 0.0046 &   152 \\
\midrule 
          web-Google & 0.0006 &    59 & 0.0008 &   234 & 0.0006 &    59 & 0.0006 &    59 & 0.0006 &    59 \\
\midrule 
     %rand-hyper-4000 & 0.0909 &     4 & 0.1111 &     4 & 0.1538 &     9 & 0.0909 &     4 & 0.0909 &     4 \\
   ff-0.4 & 0.0286 &     9 & 0.0004 &   539 & 0.0004 &   539 & 0.0004 &   539 & 0.0004 &   539 \\
  ff-0.49 & 0.0222 &     9 & 0.0067 &    24 & 0.0067 &    24 & 0.0067 &    24 & 0.0105 &    20 \\
\bottomrule
 \end{tabularx}

\end{table*}

By-and-large, the Fiedler cut, personalized PageRank, whiskers,
and \algn{metis} all tend to identify similar communities as the best.
There are sometimes small differences.  An example of a large
difference is in the Penn94 graph, where the Fiedler community
is much larger than the best PageRank community \emph{and} it has
better conductance.  In this comparison, the neighborhood
communities fare poorly.  When they identify a set of
conductance that's as good as the rest, then it is always a whisker
as well.  In the following full section,
we explore using these neighborhood communities as \emph{seeds}
for the PageRank algorithms.  This will let us take 
advantage of the observation that the neighborhood
communities reflect the \emph{shape} of the network
community plot with PageRank communities

\subsection{Empirical Core Communities}

In our theoretical work, we found that large
$k$-cores should always exist in these networks.
These should also look like good communities
and we briefly investigate this idea in
\Fig{cores}.  The standard procedure for computing
$k$-cores is to iteratively remove in degree-sorted
order using a bucket sort~\cite{batagelj2003-cores}.
We additionally store the \emph{step} when each
vertex was removed from the graph.  We sweep over
all cuts induced by this ordering, and for each
$k$-core, store the best conductance community.  These
are plotted in a line that runs from core $1$ to
the largest core in the graph.  The $1$ core is usually
large and a bad-community.  Thus, the line usually starts
towards the upper-right of each network community
plot.  Large cores are actually rather good communities.
Their conductance scores are noticeably higher
than the PageRank communities, but the network
plots seem to have similar shapes.  We'll exploit
this property in the next section.

\begin{figure}
 \begin{tabular}{@{}c@{}c@{}}
Penn94 & dblp\\[-0.5ex]
\includegraphics[width=0.5\linewidth]{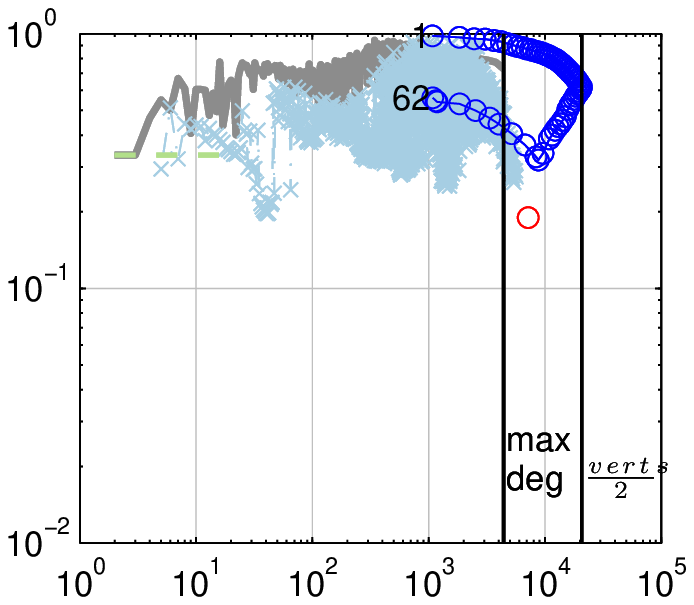} 
& \includegraphics[width=0.5\linewidth]{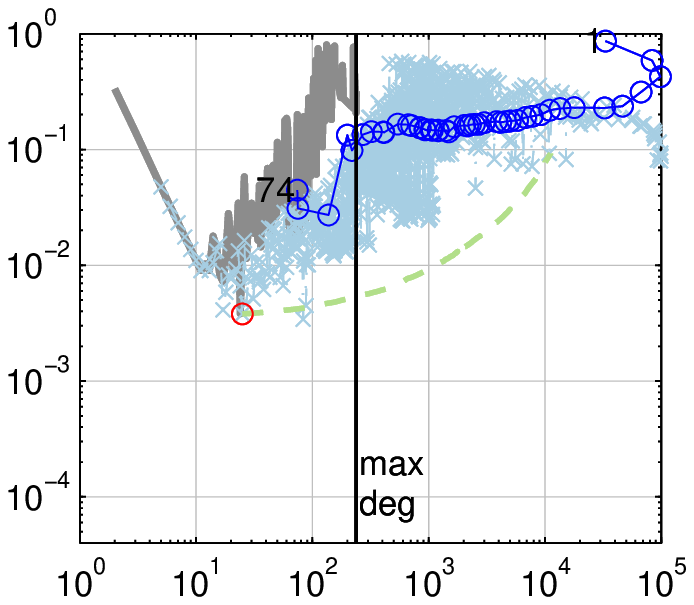} \\
\end{tabular}
\caption{Network community plots with neighborhood
communities (gray), PageRank communities (light
blue), whiskers (green) and k-cores (dark blue).}
\label{fig:cores}
\end{figure}

%
%
%Things to note: in the Forest Fire graphs, the large neighborhood
%communities have better conductance than the personalized PageRank
%communities.  The core communities are just as good, if not better.
%Whisker communities are much superior (in terms of conductance)
%to either of these groups, however.  

\section{Seeded communities} \label{sec:seeds}

Many of the theorems about extracting local communities
from seed sets~\cite{andersen2006-communities,andersen2006-local}
require that the seed set itself be a good community.  This
is precisely what our theoretical results justify for 
neighborhood communities.  Consequently, in this section,
we look at \emph{growing} the neighborhood communities
using the local personalized PageRank community algorithm
from a set of carefully chosen seeds.

One of the key problems with using the personalized PageRank
community algorithms is that finding a good set of seeds is
not easy.  For example, \cite{Gargi-2011-youtube}
describes a way to do this using the most popular videos 
on YouTube.  Such a meaningful heuristic is not always available.
We begin this section by 
empirically showing that there is an easy-to-identify set of neighborhood
communities that are local extrema in the network community plot
of the neighborhood communities.

%\setlength{\columnsep}{1em}
%\begin{wrapfigure}{r}{0.5\linewidth}
\begin{figure}[t]
\centering
 %\vspace{-\baselineskip}
 \includegraphics[width=0.75\linewidth]{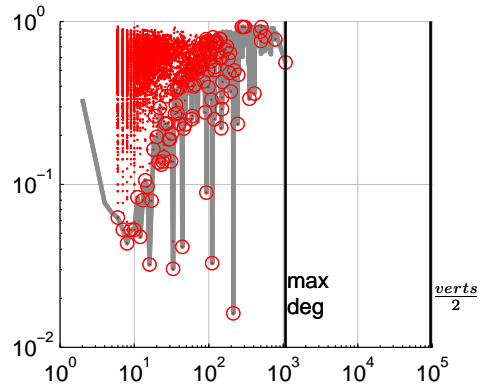}
 \caption{The conductance of locally minimal communities 
in the itdk0304 graph (red). Note that these capture
most of the local minima (downward spikes) in the profile.}
%\vspace{-2em}
\label{fig:localmin}
%\end{wrapfigure}
\end{figure}
\begin{figure}
\centering
\begin{tabular}{@{}c@{}c@{}}
\includegraphics[width=0.5\linewidth]{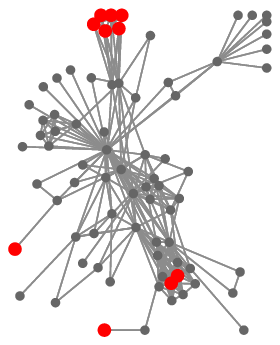} 
& \includegraphics[width=0.5\linewidth]{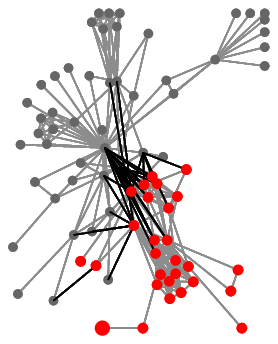} \\
& size=28, cut=31, $\phi$=0.12
\end{tabular}
\caption{(Left) The center vertices of the locally minimal vertex neighborhoods in the
Les Mis\'{e}rables are marked in red.  (Right) The best pageRank grown community
from these vertices matches the best from any seed.}
\label{fig:lesmis-pprgrow}
\end{figure}
First, some quick terminology: we say a neighborhood community is a local minima,
or locally minimal, if the conductance of the neighborhood of a vertex is smaller than the
conductance of any of the adjacency neighborhood communities.  Formally,
\[ \begin{aligned} & \phi(\ball{1}{v}) \le \phi(\ball{1}{w}) \\ & \quad \text{ for all $w$ adjacent to $v$ }  \end{aligned} \]
is true for any locally minimal communities.
We find there are only a small
set of locally minimal communities with more than 6 vertices.  
Shown in \Fig{localmin} are the conductance and sizes of 
the roughly $7000$ communities identified
by this measure for the itdk0304 graph.
Indeed, among all of the graphs with at least $85,000$
vertices, this heuristic picks out about 3\% of the vertices as local minima. In the
worst case, it picked out $100,000$ seeds for soc-LiveJournal1.
Increasing the minimum size to $10$ vertices reduces this down to $50,000$
seeds.
We then use these locally minimal neighborhoods as seed sets for the personalized PageRank
community detection procedure. Each locally minimal
neighborhoods is grown by up to 50-times its volume by solving for communities
using various values of $\sigma$ up to 50.  We also explore growing
the $k$-cores by up to $5$ times their volume.  See \Fig{lesmis-pprgrow}
for the locally minimal communities and the best grown community
from the Les Mis\'{e}rables graph.

\Fig{pprgrow} shows the results.  In these figures, we leave the baseline
neighborhood communities in for comparison. The key insight is that
the dark black line closely tracks the the outline of the pure-PageRank
based community profile.  That profile was computed by using every
vertex in the graph as a seed (although, some vertices were skipped
after 10 other clusters had already visited that vertex).  This effect
is most clearly illustrated by the email-Enron dataset.  The dark black
line identifies almost all of the local minima from the full PageRank
sweep (there are a few it misses).  A weakness of these minimal
seeds for PageRank is that they may not capture the \emph{largest} 
communities.  However, the $k$-core grown communities do seem
to capture this region of the profile (e.g.\ arxiv), although ca-AstroPh
is an exception.  

\begin{figure}
\begin{tabular}{@{}c@{}c@{}}
ca-AstroPh & email-Enron\\[-0.5ex]
\includegraphics[width=0.5\linewidth]{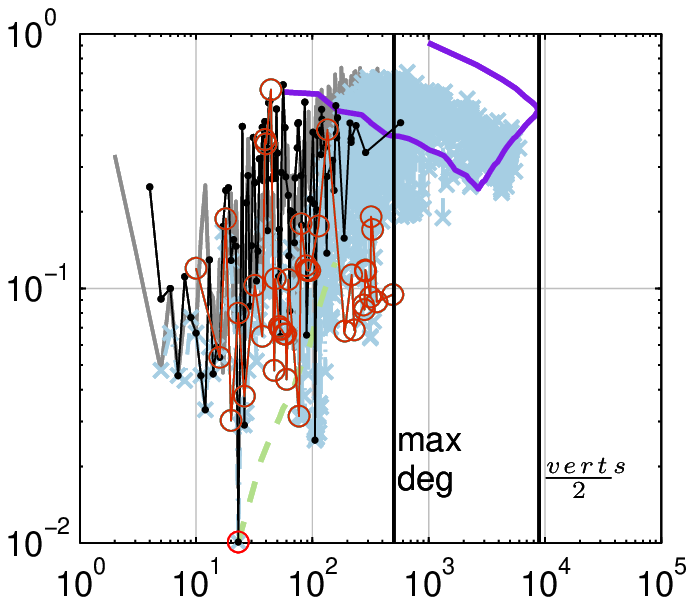} 
& \includegraphics[width=0.5\linewidth]{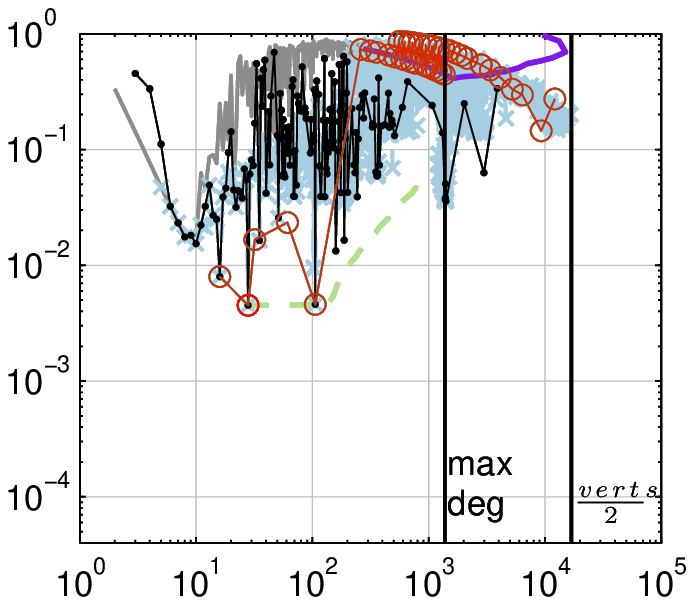} \\
arxiv & fb-A\\[-0.5ex]
\includegraphics[width=0.5\linewidth]{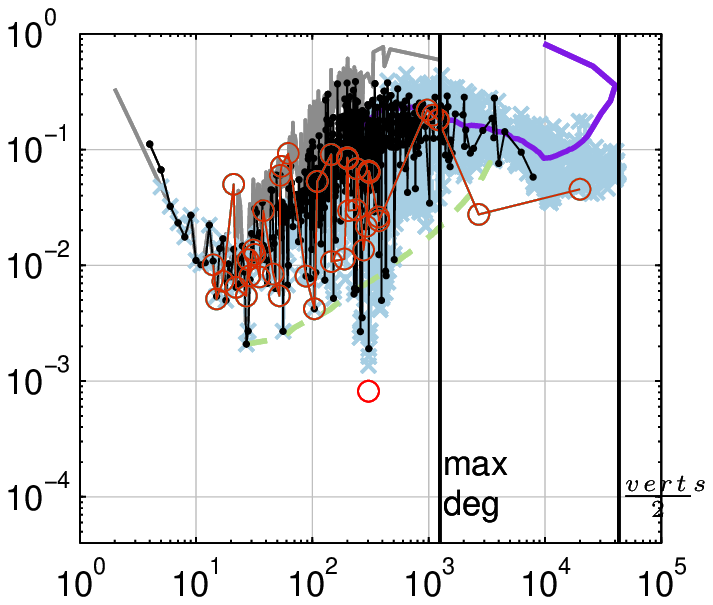} 
& \includegraphics[width=0.5\linewidth]{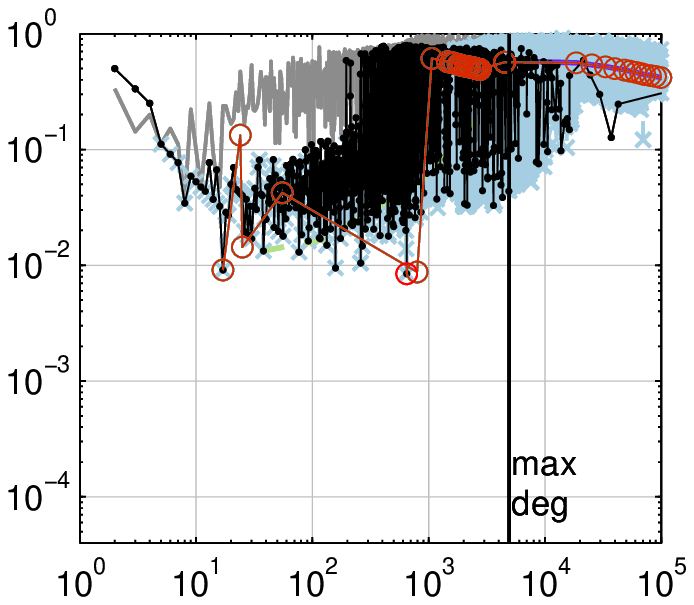} \\
\end{tabular}
\caption{Network community plots with neighborhood
communities (gray), PageRank communities (light
blue), whiskers (green), k-cores (purple), locally
minimal seed PageRank communities (black), and
k-core seeded PageRank communities (red).}
\label{fig:pprgrow}
\end{figure}

% The theory of seeded random cuts requires starting nearby a good cluster.
% Here, we use the following heuristic to identify good 
% seed sets.  

% Show that we can find the extrema in the graphs by using
% the neighborhood min

\section{Concluding discussions}

%% To discuss, is a skew-degree sufficient?

We recap.  Community detection is the problem of
finding cohesive collections of nodes in a network.
We formalize this as finding vertex sets with
small conductance.  Modern information networks have
many distinctive properties, including a large clustering
coefficient and a heavy-tailed degree distribution.  
We derive a set of theoretical results that show these
properties imply that such networks will have vertex
neighborhoods that are \emph{themselves} sets of small
conductance.  Although our theoretical bounds are weak, they
suggest the following experiment: measure the conductance
of vertex neighborhoods.

Algorithms to compute \emph{all such conductance scores}
are easy to implement by 
modifying a routine for computing local clustering coefficients.
We evaluate these communities on a set of real-world networks.
In summary, our results support the idea that there are
many neighborhood communities which are \emph{good communities}
in a conductance sense.   
They may be smaller than desired, however.

We next investigate finding a set of locally minimal 
communities.  These communities represent the 
\emph{best of the neighborhood}.  We find that these locally
minimal communities, of which there are many fewer than
vertices in the graph (usually around 3\%), capture the
local minimal in the network community profile plot.  
More importantly, they  
can be enlarged using a local personalized PageRank 
community detection procedure.  Afterwards, the profile
of these ``grown'' neighborhoods is strikingly close 
to the profile of the PageRank communities when seeded
with all vertices individually.  While we do not discuss
timing due to the variability in the quality of implementations,
this later procedure is much faster in our experiments.

These findings have implications for future studies
in community detection.  One explanation for the results
with the PageRank seeds is that vertex neighborhoods form 
the core of \emph{any} good community in the network.
We highlight this as a direction for future research into
neighborhood communities.  

\bibliographystyle{abbrv}
\bibliography{bib}  

\end{document}